\documentclass{article}          
\usepackage{psfrag}
\usepackage{algorithm}
\usepackage[noend]{algorithmic}
\algsetup{indent=2em}
\usepackage{amsmath}
\usepackage{graphicx,psfrag,epsf}
\usepackage{upgreek}
\usepackage{wrapfig}
\usepackage{boxedminipage}
\usepackage{pifont}
\usepackage{euscript}
\usepackage{dcolumn}
\usepackage{float}
\usepackage{fancybox}
\usepackage{rotating}
\usepackage{multirow}
\usepackage{latexsym}
\usepackage{amssymb}
\usepackage{txfonts}
\usepackage{amsmath}
\usepackage{epsfig}
\usepackage{color}
\usepackage{verbatim}
\usepackage{subfigure}
\usepackage{mathrsfs}
\usepackage{url}
\usepackage{lscape}
\usepackage{graphicx}
\usepackage{amsmath}
\usepackage{amssymb}
\usepackage{amsfonts}
\usepackage{euscript}
\usepackage{enumitem}
\usepackage{color}
\usepackage{siunitx}
\newtheorem{theorem}{Theorem}[section]
\newtheorem{lemma}{Lemma}[section]
\newenvironment{proof}{{\em Proof.}}{\hspace*{\fill}$\Box$}

\newcommand{\m}[1]{\mathbf{#1} }
\renewcommand{\v}[1]{\boldsymbol{#1}}
\newcommand{\bb}[1]{\mathbb{#1}}

\renewcommand{\c}[1]{\mathcal{#1}}

\newcommand{\di}{\mathrm{d}}
\newcommand{\cas}{\stackrel{\mathrm{a.s.}}{\longrightarrow}}
\newcommand{\gvn}{\, |\,}

\DeclareMathOperator*{\argmin}{argmin}
\DeclareMathOperator*{\argmax}{argmax}
\begin{document}

\title{Regenerative Simulation for the Bayesian Lasso}

\author{Y.-L. Chen \& Z. I. Botev \\ UNSW Sydney, Australia }

\maketitle

\begin{abstract}
The Gibbs sampler of Park and Casella 
is one of the most popular MCMC methods for
 sampling from the posterior density of the Bayesian Lasso regression. As with many  Markov chain samplers, their Gibbs sampler  lacks a theoretically sound   method of output analysis --- a method for estimating the variance of a given ergodic average and estimating how closely the chain is sampling from the stationary distribution, that is, the burn-in.

In this paper, we address this shortcoming by identifying   regenerative structure in the 
 sampler of Park and Casella, thus providing a 
theoretically sound method of assessing its 
performance. The 
regenerative structure provides both a strongly consistent variance estimator, and an estimator of (an upper bound on) the 
total variation distance from the target posterior density. The result is a simple and theoretically sound way to assess the stationarity of the  Park and Casella  and, more generally, other MCMC samplers, for which regenerative simulation is possible.

We  perform a numerical study in which we validate the standard errors calculated by our regenerative method by comparing it with the standard errors calculated by an AR(1)
heuristic approximation. Thus, we show that for the  Bayesian Lasso model, the regenerative method is a viable and theoretically justified alternative to the existing ad-hoc MCMC convergence diagnostics.

\end{abstract}

\section{Introduction}
\label{intro}
The  linear Lasso regression and its Bayesian analogue are studied extensively and have appealed to many practitioners \cite{Park2008bayesian}. Inference for the Bayesian Lasso  requires one to take expectations with respect to $\pi$, the posterior density. These expectations are intractable and call for Monte Carlo statistical methods such as Markov chain Monte Carlo (MCMC).

The idea is to construct a Markov chain $\{\v X_0,\ldots,\v X_t\}$ starting from some (possibly random) initial state $\v X_0$, with invariant density $\pi$, so that the average of the sample path converges to the expectation one wishes to compute. Denoting $\bb E_\pi h$ the expectation of $h$ with respect to $\pi$, we have under suitable conditions:
\begin{equation}
\bar h_t:=\frac{1}{t+1}\sum_{k=0}^{t}h(\v X_k)\cas \bb E_\pi h,\quad t\rightarrow\infty.
\label{eq:MeanEstimate}
\end{equation}
 One of the most popular MCMC samplers for the Bayesian Lasso is the Gibbs sampler of Park and Casella \cite{Park2008bayesian}.
Despite its wide use,  the sampler still lacks a systematic way to: (i) estimate the variability of the estimator $\bar h_t$; (ii) assess how closely (in total variation distance) a state of the Markov chain follows the target posterior (this is related to the problem of  estimating the size of the burn-in of the Markov chain).

Currently one resorts to heuristic approximations to address both (i) and (ii). 
For example, a popular approach to address (i) is the \emph{batch means} variance estimator to estimate the standard error of $\bar h_t$. The batch means variance estimator requires covariance stationarity  \cite{law1984confidence1}, which is difficult to verify in practice. Furthermore, for the batch means estimator to be consistent, each batch size has to diverge to infinity and in practice it is not clear how large each batch has to be.

There are many existing works that address (ii), but, roughly speaking, these approaches can be categorized into two groups. The first approach is to analyze the transition kernel of the Markov chain and construct a total variation distance bound between the transition density and the invariant density (see \cite{meyn1994computable,rosenthal1995minorization}, for example). Despite its theoretical soundness, this approach  often requires difficult or intractable  analysis. 

The simpler and more popular alternative is to examine the output of the Markov chain sampler. These approaches are known as  ``convergence diagnostics" in the literature, and include examining the decay of the sample autocorrelation plots \cite{polson2014bayesian} or running multiple chains until the chains roughly stay in the same region of the state space (for example in the popular Bayesian inference software WinBugs). These heuristics or rules-of-thumb mostly provide a pictorial convergence assessment and rarely a quantitative one. Indeed, Cowles et al. \cite{cowles1996markov} mention that ``...statisticians rely heavily on such diagnostics, if for no other reason than that a weak diagnostic is better than no diagnostic at all." 

 In this paper, we address both problems (i) and (ii)  by identifying the \emph{regenerative structure} in the output of the Park \& Casella Gibbs sampler.
 Regenerative simulation is a compromise between the two extremes above (analytical bounds and diagnostic plots) --- it
relies both on some preliminary analytical work and  on the output of the MCMC sampler.  Roughly speaking, given the Markov chain 
$\{\v X_0,\ldots,$ $\v X_t\}$, with invariant density $\pi$, the aim is to identify the times where the process stochastically `restarts' itself, thereby breaking the chain into iid segments.  Our novel approach  uses results from \cite{glynn1994some} to construct a total variation distance bound between the distribution of $\v X_t$ and the invariant density, and then  uses the (regenerative) iid output from the sampler to estimate the  unknown constants in this bound. 
In short, we demonstrate that a regenerative structure is all that is needed to address both (i) and (ii). We note that while the idea of using regeneration to address (i) goes back to \cite{mykland1995regeneration,jones2001honest},  these works do not address the important burn-in issue of (ii) via regeneration.  

In summary, our contribution is twofold: 1) to apply the regenerative method  \cite{mykland1995regeneration,jones2001honest} to the Park and Casella Gibbs sampler and address (i); and 2) to show how regenerative simulation can address the burn-in issue (ii) for any MCMC sampler, not just for the specific sampler of Park \& Casella.

The rest of the paper is structured as follows. In Section~\ref{sec:reg process}, we provide  background on regenerative simulation and then  discuss how regeneration can address the problem of MCMC burn-in, that is, issue (ii). Then, in Section~\ref{sec:nummelin} we show how regenerative simulation can be applied to the Park \& Casella Gibbs sampler. This is followed by a numeric section that uses two popular test cases, where we compare the regenerative estimators with the estimators based on the 
$AR(1)$ heuristic approximation. Finally, we draw conclusions on the benefits of regenerative simulation for addressing both issues (i) and (ii).

\section{Convergence Assessment for Regenerative Processes}
\label{sec:reg process}
Before presenting our novel contribution, we briefly summarize known facts about regenerative processes. 
 Recall that  $\{\v X_k, k=0,1,2,\ldots\}$ is said to be
a \emph{zero-delayed discrete-time  regenerative} process if there exist times
\[0=T_0 \leq T_1 \leq T_2 \leq T_3 \leq \ldots\]
such that $
\{\v X_{T_r+k} , 0\leq k \leq T_{r+1}-1\}
$
and
$\{\v X_{T_s+k} ,
0\leq k \geq T_{s+1}-1\}
$ are iid for all $r\not=s$. As a consequence, the lengths of the tours or cycles
\[
M_r=T_{r+1}-T_r,\qquad r=1,2,\ldots
\] 
are iid.   Suppose $h$ is a measurable function with $\bb E|h(\v X_k)|<\infty$ and 
\[
H_r:= \sum_{k=T_{r-1}}^{T_r-1} h(\v X_k)
\]
  Then, we know \cite{cinlar2013introduction} that $\v X_k$ converges in distribution
to a random variable $\v X\sim \pi$ such that
\[
\bb E_\pi h(\v X)=\frac{\bb E \sum_{k=0}^{T_1-1}h(\v X_k)}{\bb E M_1}=\frac{\bb E H_1}{\bb EM_1}
\]
We denote the distribution of this $\v X$ as $\pi$. It is the stationary distribution of the 
regenerative process. We also have \cite{zheng2016extensions}:
\[
\hat q_t:= \frac{1}{t}\sum_{k=0}^{t-1}h(\v X_k)\stackrel{\mathrm{a.s.}}{\longrightarrow} \bb E_\pi h=:q
\quad \textrm{and}\quad
\sqrt{t}(\hat q_t-q)\stackrel{\mathrm{d}}{\longrightarrow} \mathsf{N}(0,\gamma^2),
\]
where  $\gamma^2$ is the so-called \emph{time-average variance constant} (TAVC). In fact, the TAVC is asymptotically the same as the mean squared error of $\sqrt{t}\hat q_t$. Note that,  the regenerative process may or may not be Markovian. If it is Markovian, then we have a Markov chain with stationary and limiting distribution $\pi$.

\subsection{Regenerative Mean Square Error Estimator}
With a regenerative process, such as the above, 
it is well-known \cite{glynn1994some,zheng2016extensions} how to  estimate the TAVC using the ratio estimator:
\begin{equation}
\label{TAVC}
\hat\gamma^2_t=\frac{\sum_{r=1}^{N(t)}(H_r-\hat q_t M_r)^2}{t},
\end{equation}
where $N(t)=\max\{n: T_n\leq t\}$ is the number of regenerations
that have occurred after running the process for $t$ steps.

Arguably the simplest and most frequently used alternative to \eqref{TAVC} is the \emph{batch means estimator}. It is applied when the process under consideration is a Markov process and identifying the regeneration events is not possible.

 The batch means estimator divides a single run of a Markov chain, $\{\v X_1,\ldots, \v X_t\}$ into $n$ `batches' of $m$ adjacent observations (so that $t=m\times n$). Denoting the sample mean of the $m$ observations from the $k$-th `batch' by $\v{\tilde X}_k$,  the batch means variance estimator is given by \cite{law1979sequential}
\[\textstyle
\hat\gamma^2_\textrm{batch}=\frac{1}{n-1}\sum_{k=1}^n\left(\v{\tilde X}_k-\frac{1}{m\times n}\sum_{k=1}^{m\times n}\v X_k\right)^2.
\] 
The batch means variance estimator is motivated by the  fact that the dependence  between  adjacent batch means goes down to zero as $m\rightarrow\infty$ (see \cite{law1979sequential} for more details).  
For this reason, \cite{jones2001honest} views the batch means estimator as an ad-hoc variant  of  the regenerative estimator \eqref{TAVC}.

Unfortunately, ensuring the consistency of  $\hat\gamma^2_\textrm{batch}$ is nontrivial.  On the one hand \cite{damerdji1994strong} shows that if $n\rightarrow\infty$ and $m\rightarrow \infty$, then the batch variance estimator is consistent. On the other hand, \cite{glynn1991estimating} shows that for any fixed $n$ and $m\rightarrow\infty$, any  batch means estimator of a stationary quantity of interest is \emph{not} consistent.  Thus,  consistency requires that both $n$ and $m$ grow without bound.  In practice, however, the lack of independence between batches makes it very difficult to determine how large $n$ needs to be and how $m$ needs to grow as a function of  $n$.   In contrast, if regeneration events can be identified within the Markov chain, then the iid regenerative structure  ensures that such difficulties do not exist.

\subsection{Novel regenerative  burn-in estimator}
In this section, we explain how to estimate the total variation discrepancy of a Markov chain for which we can identify
its regenerative events. 

Recall that the total variation distance between $\kappa_t(\cdot|\v x_0)$, the $t$-th step transition kernel of a Markov chain starting at $\v x_0$, and the invariant density $\pi$ is defined by 
\begin{equation*}
\|\kappa_t(\cdot|\v x_0)-\pi\|_{\mathrm{TV}}=\sup_{A\in\mathscr B}|\kappa_t(A|\v x_0)-\pi(A)|,
\end{equation*}
where $\mathscr B$ is the Borel $\sigma$-algebra (and henceforth omitted from the notation). Also, if $\|\kappa_t(\cdot|\v x_0)-\pi\|_{\mathrm{TV}}\leq c_1\exp(-\varepsilon t)$ for some $\varepsilon>0$
and constant $c_1$ (possibly depending on $\v x_0$),
then the underlying Markov chain is said to be geometrically ergodic.  

Now, suppose that we initialize the chain from a random initial $\v X_0$ drawn from some density $\pi_0$. Then, the $t$-step transition kernel is obtained by taking the expectation with respect to $\v X_0$, namely,
$\bb E[\kappa_t(A\gvn\v X_0)]$.
 We define the $\epsilon$-\emph{burn-in} of a Markov chain with transition kernel $\kappa$ as the smallest $t$ for which
$
\|\bb E[\kappa_t(\cdot\gvn\v X_0)]-\pi\|_{\mathrm{TV}}<\epsilon
$, that is:
\[
t_\textrm{b}:=\min\{t:\|\bb E[\kappa_t(\cdot\gvn\v X_0)]-\pi\|_{\mathrm{TV}}<\epsilon\}
\]
Hence, a theoretically sound assessment of convergence, is to construct an estimate of (or a bound for) $\|\bb E[\kappa_t(\cdot\gvn\v X_0)]-\pi\|_{\mathrm{TV}}$, and  examine how fast it decays with respect to $t$. 
Since a simple analytical formula is too difficult to derive, practitioners turn to heuristics such as examining the  autocorrelation plots (mentioned in the introduction) or experimenting with the Markov chain using different starting values, $\v X_0$. 

Instead, we adopt a more theoretically sound approach that is a compromise between the extremes of an exact theoretical bound and an heuristic diagnostic plot.
Our key insight is that the bias properties of regenerative estimators \cite{glynn1994some}  allow us to bound the total variation distance, as follows.
\begin{theorem}[Total Variation Bound for MCMC]
\label{th:converge}
Let $\kappa_t(\cdot|\v X_0)$ with $\v X_0\sim\pi_0$ be the $t$-step transition kernel of a geometrically ergodic Markov chain  with  invariant density $\pi$. 
Suppose we can identify  regenerative times of the Markov chain and  assume that $\v X_0\sim \pi_0$ initialized a new regenerative cycle for simplicity.  Then, we have (for some constant $\varepsilon>0$)
\[
\|\bb E[\kappa_t(\cdot|\v X_0)]-\pi\|_\mathrm{TV}\leq \frac{\eta}{t+1}+\c O(\exp(-\varepsilon t)),
\]
where $\eta=\frac{\bb E M^2_1-\bb E M_1}{2\bb E M_1}$ with  $M_1,M_2,\ldots$
denoting the iid regenerative cycles.
\end{theorem}
The proof is given in the Appendix.

A key insight from  the theorem above  is  that an asymptotic upper bound for the $\epsilon$-burn-in, $t_\textrm{b}$, is
$
\left\lceil \eta/\epsilon\right \rceil,
$
and that the constant $\eta$ can be estimated from simulation using the iid realizations of $M_1,M_2,\ldots$. 

In summary, ($T_n=M_1+\cdots+M_n$  and $N(t)=\max\{n: T_n\leq t\}$) our novel  estimator for the $\epsilon$-burn-in  is:
\begin{equation}
\label{eps-burn-in}
\left\lceil\frac{\sum_{k=1}^{N(t)}M_k^2-\sum_{k=1}^{N(t)}M_k}{2\epsilon\sum_{k=1}^{N(t)}M_k }\right \rceil\;.
\end{equation}
This estimator can admittedly be quite conservative as it relies on an upper bound of the total variation distance, not on the actual distance. 

The following table summarizes the current popular practice and our suggested alternative.
{\scriptsize
\begin{center}
\begin{tabular}{|c|c|c|}
\hline
\textbf{Issue} & (i) \emph{Estimate MSE} & (ii) \emph{Assessing the convergence}\\
\hline
\textbf{Theoretical approach} & Compute/Bound TAVC& Compute/bound the TV distance\\
\hline
\textbf{Popular  approach} & Batch-means estimator & Diagnostic plots\\
\hline
\textbf{Regenerative approach} &TAVC Estimator \eqref{TAVC} &  Bias Estimator \eqref{eps-burn-in}\\
\hline
\end{tabular}
\end{center}
}

In the next section, we apply the variance estimator \eqref{TAVC}
and the $\epsilon$-burn-in estimator \eqref{eps-burn-in} to the 
Gibbs sampler of Park \& Casella. Note that their sampler  is known  to be geometrically ergodic \cite{khare2013geometric}, so that the results in Theorem~\ref{th:converge} apply. 


\section{Regenerative Simulation for Park \& Casella Sampler}
\label{sec:nummelin}
In order to assess  the convergence  of the Park \& Casella sampler via the regenerative estimators \eqref{TAVC} and \eqref{eps-burn-in},
we first need  to identify the regeneration events in the output of the sampler. The most common method for identifying  regenerative structure in Markov chains is the  state-space augmentation method of Nummelin \& Mykland \cite{nummelin1984,mykland1995regeneration,jones2001honest}.

\subsection{Nummelin state-space augmentation}

To identify regenerative structure in a Markov chain with transition kernel $\kappa(\v x_{k+1}|\v x_k)$ and invariant density $\pi$, we first need to 
establish the so-called \emph{minorization condition}. Namely, 
 we seek  a probability measure $\nu$ and a function $s$ such that
\begin{equation}
\kappa(\v x_{k+1}|\v x_{k})\geq \nu(\di \v x_{k+1}) s(\v x_k),\quad \forall \v x_k.
\label{eq:minorization}
\end{equation}
Once \eqref{eq:minorization} is established, one can then simulate the Markov chain $\v X_1,\v X_2,\ldots$ via the mixture representation of $\kappa$:
\begin{equation}
\textstyle
\kappa(\v x_{k+1}|\v x_k)=s(\v x_k)\nu(\di \v x_{k+1})+(1-s(\v x_k))
\frac{\kappa(\v x_{k+1}|\v x_k)-\nu(\di \v x_{k+1}) s(\v x_k)}{1-s(\v x_k)}.
\label{eq:mixture}
\end{equation}
Thus, a regenerative structure arises in this process, because  $\{\v X_1,\ldots,\v X_{k}\}$ is independent of 
$\{\v X_{k+1},\v X_{k+2},\ldots\}$ whenever $\v X_{k+1}$ is simulated from the first component, $\nu(\di \v x_{k+1})$, of the mixture.

Simulation from the mixture components of  \eqref{eq:mixture}  may be difficult, if not impossible. Indeed, an important insight of \cite{mykland1995regeneration} is that one does not need to simulate from the mixture densities of \eqref{eq:mixture} directly. Instead, one can simulate from $\kappa(\v x_{k+1}|\v x_k)$ in the usual manner, and identify regeneration times \emph{retrospectively}. To be precise, given $\v x_k$, the $k$-th realization, we can simulate $\v X_{k+1}$ from $\kappa(\v x_{k+1}|\v x_k)$ and decide that regeneration
has occurred with retrospective probability:
\[
\psi_k:=\bb P[\textrm{ regeneration at } k \;|\,\v X_k,\v X_{k+1}]=\frac{s(\v X_k)\nu(\v X_{k+1})}{\kappa(\v X_{k+1}|\v X_k)}.
\]
That is to say, if one wishes to incorporate regeneration in a geometrically ergodic MCMC sampler, one  proceeds as follows.
\begin{enumerate}
 \item Establish \eqref{eq:minorization} for the transition density of the MCMC sampler.
 \item  Simulate the Markov chain $\{\v X_1,\v X_2,\ldots,\v X_t\}$ as usual (e.g., running the Gibbs sampler of Park \& Casella), starting from $\v X_0$.
 \item For $k\in\{1,\ldots,t-1\}$, simulate a Bernoulli random variable  with success probability  $\psi_k$ to decide whether regeneration has occurred.
\end{enumerate}
In the next section we establish the minorization condition for the Gibbs sampler of Park \& Casella and provide a formula for $\psi_k$.
In this way, we will have all the ingredients to run the above algorithm.

\subsection{Application to Park \& Casella sampler}
Given the response variable $\v Y$ and model matrix $\m X$,
the hierarchical formulation of Bayesian Lasso linear regression model is as follows (here $\v\beta,\sigma$ are model parameters and $\lambda$  is the Lasso regularization parameter):
\[
\begin{split}
\beta_j|\lambda&\overset{i.i.d}{\sim}\mathsf{Laplace}(0,1/\lambda),\quad\mbox{for }j\in\{1,\ldots,p\}\\
\v Y|(\v\beta,\lambda,\sigma^2)&\sim\mathsf{N}(\m X\v\beta,\sigma^2\m I).
\end{split}
\]
It follows that inference for the Bayesian Lasso linear regression requires one to take expectations with
respect to the posterior density  (for  simplicity of notation we drop $\v y$)
\begin{equation}
 \pi(\v\beta|\lambda,\sigma^2)=\frac{(\lambda/2)^p\exp\left(-\frac{1}{2\sigma^2}\|\v y-\m X\v\beta\|_2^2-\lambda\|\v\beta\|_1\right)}{\ell(\lambda,\sigma^2)},
\label{eq:posterior}
\end{equation}
where $\ell(\lambda,\sigma^2):=\int (\lambda/2)^p\exp\left(-\frac{1}{2\sigma^2}\|\v y-\m X\v\beta\|_2^2-\lambda\|\v\beta\|_1\right) \di \v\beta$
is the marginal likelihood of the pair $(\lambda,\sigma^2)$.

Recall (see Appendix~\ref{app:BLasso} for details or \cite{Park2008bayesian}) that the transition density for the Gibbs sampler of Park \& Casella is
\[
\kappa(\underbrace{(\v\beta_{k+1},\v \tau_{k+1})}_{\v x_{k+1}}|\underbrace{(\v \beta_k,\v\tau_k)}_{\v x_k})=\pi(\v\tau_{k+1}|\v\beta_{k})\pi(\v\beta_{k+1}|\v\tau_{k+1}),
\]
where 
$
\pi(\v\tau_{k+1}|\v\beta_{k})
$ is the joint density of independent
$\mathsf{Wald}(\lambda',\mu_j')$ random variables with  $\lambda'=\lambda^2$ and $\mu_j'=\lambda/|\beta_j|$ (see, for example, \cite{chhikara1988}) and $\pi(\v\beta_{k+1}|\v\tau_{k+1})$ is the density of the multivariate $\mathsf{N}(\m A\m X^\top \v y, \sigma^2\m A)$ distribution, where $\m A:=\m X^\top\m X+\mathrm{diag}(\v\tau)$. We have the following lemma whose proof is  in the Appendix.
\begin{lemma}[Regenerative conditions for  Park \& Casella  sampler]~\\
\label{lem:minorization}
 Let $\hat{\v \beta}$ be the solution to the frequentist Lasso penalized regression model:
\[
\hat{\v \beta}=\argmin_{\v\beta}\{ \|\v y-\m X\v\beta\|^2_2+\lambda\|\v\beta\|_1\}
\]
 and let $\c D=\bb R^p\times [\v c,\v d]$ be a  subset of $\bb R^p\times\bb R_+^p$, the state space on which $(\v\beta,\v\tau)$ is defined. Define the probability measure $\nu(\v\beta_{k+1},\v\tau_{k+1})$: 
\begin{equation}
\label{fresh reg.}
\nu(\v\beta_{k+1},\v\tau_{k+1})=\varepsilon^{-1}\kappa((\v\beta_{k+1},\v \tau_{k+1})|(\hat{\v\beta},\v 1))\times\bb I\{(\v\beta_{k+1},\v\tau_{k+1})\in \c D\}\;,
\end{equation}
where $\v 1\in \bb R^p$ is the vector of ones and
$\varepsilon$ is the normalizing constant for $\nu$.
Let the notation $\v a_+$ mean setting all negative entries of the vector $\v a$ to zero, and similarly $\v a_-$ sets all positive entries of $\v a$ zero ($\v a^{2}$ means squaring each entry). Then, the measure $\nu$ and the function:
\[\textstyle
s(\v\beta_k,\v\tau_k)=\varepsilon\exp\left(-\frac{1}{2}\v d^\top\v(\v\beta_k^2-\hat{\v\beta}^2)_+-\frac{1}{2}\v c^\top\v(\v\beta_k^2-\hat{\v\beta}^2)_-\right)
\]
satisfy the minorization condition:
\[
\kappa((\v\beta_{k+1},\v \tau_{k+1})|(\v \beta_k,\v\tau_k))\geq\nu(\v\beta_{k+1},\v\tau_{k+1})s(\v\beta_k,\v\tau_k),\qquad\forall (\v\beta_k,\v\tau_k).
\]
Conditional on the simulated states $(\v\beta_{k},\v\tau_k)$ and
$(\v\beta_{k+1},\v\tau_{k+1})$,
the probability that a regeneration at the $k$-th step has occurred is:
\begin{equation}
\begin{split}
	\psi_k&=\textstyle\exp\left(-\frac{(\v d-\v\tau_{k+1})^\top(\v\beta_k^2-\hat{\v\beta}^2)_+}{2}-\frac{(\v c-\v\tau_{k+1})^\top(\v\beta_k^2-\hat{\v\beta}^2)_-}{2}\right)\times\\
	&\qquad\times\bb I\left\{\v\tau_{k+1}\in[\v c,\v d]\right\}
	\end{split}
\label{eq:RegProb}
\end{equation}
\end{lemma}
To start the Markov chain with a fresh regenerative cycle, we need only  simulate an initial state  from $\nu(\v\beta,\tau)$ in \eqref{fresh reg.} above.
Now, we have all the ingredients for identifying regeneration events during the course of running the Gibbs sampling of Park \& Casella.


\subsection{Practical tuning of algorithm}

Our simulation experience suggests  that it does pay off to put some effort in optimizing the probability of regeneration with respect to $\v c,\v d\in\bb R_+^p$. Clearly, as the volume of the hyper rectangle $[\v c,\v d]$ in \eqref{eq:RegProb} decreases, the probability of observing $\v\tau\in [\v c,\v d]$ shrinks to zero. However, if one makes $[\v c,\v d]$ too large, then the exponential term approaches zero, that is, the probability of observing a regeneration again shrinks to zero. This suggests that  we can search for the $[\v c^*,\v d^*]$ that yield the optimal tradeoff between these two antagonistic conditions.

Ideally one should solve the optimization program 
\[
(\v c^*,\v d^*)=\argmax_{\v c , \v d}\bb E [\bb P[\textrm{ regeneration at } k \;|\,\v X_k,\v X_{k+1}]],
\]
where the expectation is with respect to  a pair $\v X_k,\v X_{k+1}$
that  is in stationarity (that is, $\v X_k\sim\pi$
and $\v X_{k+1}\sim\kappa(\v X_{k+1}|\v X_k)$). 
\begin{figure}[htb]
\begin{center}
\includegraphics[scale = 0.54]{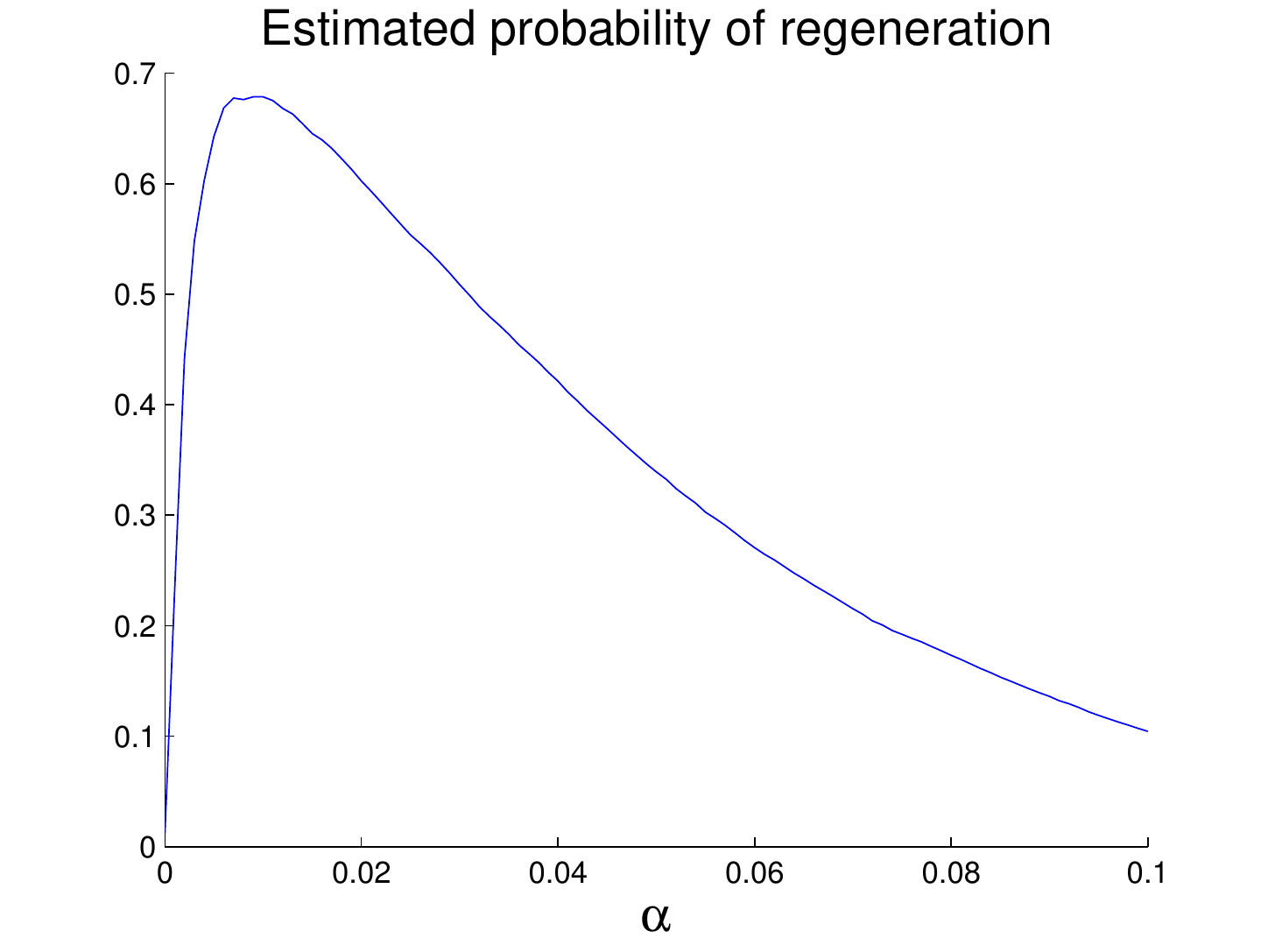}
\caption{Grid search to maximize the number of regenerations for the diabetes dataset. Here the empirical Bayes estimator for the pair $(\lambda,\sigma)$ is $(0.00431,53.5)$. We can see a clearly pronounced maximum at $\alpha\approx 0.01$.}
\label{fig:GridSearch}
\end{center}
\end{figure}

There are two difficulties here.
First, solving the program analytically is impossible. Second, the integration is $2p$-dimensional. 

Our solution to the first difficulty is to first simplify the optimization to a univariate optimization in terms of a single variable $\alpha$. More precisely, to simplify the grid search optimization, we let $\alpha\in[0,1]$ and for each $j$, we denote the lower and the upper $\alpha$-quantile for $\tau_j$ by $c_{j,(\alpha)}$ and $d_{j,(\alpha)}$ respectively. 
Approximately, we have 
\[
\bb P[\tau_j<c_{j,(\alpha)}]=\bb P[\tau_j>d_{j,(\alpha)}]=\alpha
\]
Thus, instead of solving the optimization program for a general $(\v c^*,\v d^*)$, we solve the univariate program 
\[
\alpha^*=\argmax_{\alpha}\bb E [\bb P[\textrm{ regeneration at } k \;|\,\v X_k,\v X_{k+1}]]
\]
 and use $[\v c_{(\alpha^*)},\v d_{(\alpha^*)}]$.

For the second difficulty, note that we already have access to a Gibbs sampler which can generate the sample paths quickly. 
Thus, to perform the optimization for $\alpha$, we run a pilot  of the Gibbs sampler to obtain an approximate empirical distribution for many pairs $(\v X_k,\v X_{k+1})$. We then use a grid search to maximize the estimated probability of regeneration with respect to $\alpha$. 
The procedure is summarized in the following pseudo-code.
\begin{algorithm}[H]
	\caption{: Grid search optimization for $\hat\alpha^*$}
	\begin{algorithmic}
	\REQUIRE {solution to the frequentist Lasso $\hat{\v\beta}$, grid $\alpha=(\alpha^{(1)},\alpha^{(2)},\ldots,\alpha^{(q)})$, number of samples in the pilot run $t$}
	\STATE{Obtain empirical distribution $(\v\beta_1,\v\tau_1),\ldots,(\v\beta_t,\v\tau_t)$ from the Gibbs sampler.}
	\STATE{Approximate the empirical marginal distribution for $\tau_j$ by the ordered statistics $(\tau_{j,(1)},\ldots,\tau_{j,(t)})$ for each $j\in\{1,\ldots,p\}$}
	\FOR{$i\in\{1,\ldots,q\}$}
	\STATE{Approximate $c_{j,(\alpha^{(i)})}$, and $d_{j,(\alpha^{(i)})}$ from the empirical marginal distribution for each $j$}
	\STATE{Compute $\psi_k^{(i)}$ using $c_{j,(\alpha^{(k)})}$ and $d_{j,(\alpha^{(k)})}$ for $k\in\{1,\ldots,t-1\}$ using \eqref{eq:RegProb}}
	\STATE{$\hat{\psi}^{(i)}\leftarrow\frac{1}{t-1}\sum_{k=1}^{t-1}\psi_k^{(i)}$}
	\ENDFOR
	\STATE{Choose $i^*$ such that $\psi^{(i)}$ is maximum}
	\RETURN{Return $\hat\alpha^*\leftarrow \alpha^{(i^*)}$}
	\end{algorithmic}
	\label{alg:GridSearch}
\end{algorithm}
Figure~\ref{fig:GridSearch} shows the result of the univariate grid search for the diabetes example considered in  Section~\ref{sec:numerics}.

It is important to note that the above optimization does not improve the convergence of the sampler, but simply helps identify more regenerative events (which occur even when they go unidentified). Identifying more regenerations only allows us to quantify the error in the MCMC estimate more accurately, but does nothing to speed up the convergence.

\section{Numeric Examples}
\label{sec:numerics}
To validate the regenerative 
results, we will use an $AR(1)$ process as a  heuristic approximation to the Markov chain output.
Recall that an $AR(1)$ process is given by:
\[
Y_{t+1} = c+ \rho Y_t+\varepsilon_t,\quad t=0,1,2,\ldots
\]
where $\varepsilon\sim\mathsf{N}( 0,\sigma_\varepsilon^2)$. Suppose the process starts at some initial state $Y_0=y_0$  and $|\rho|<1$. Then, the formula for the mean and
 variance  is
\[
\mu_t=c\frac{1-\rho^t}{1-\rho}+\rho^t y_0, \quad \mathrm{Var}(Y_t)=\sigma^2_\varepsilon\frac{1-\rho^{2t}}{1-\rho^2}
\]
Thus, the stationary distribution of the $AR(1)$
process is a Gaussian with mean and variance, $\mu=\lim_{t\uparrow\infty}\mu_t=c/(1-\rho)$ and $\sigma^2=\lim_{t\uparrow\infty}\mathrm{Var}(Y_t)=\frac{\sigma_\varepsilon^2}{1-\rho^2}$, respectively. 
In other words, once we have estimates for the
$AR(1)$ model parameters,  we can upper bound the total variation distance as follows.

\begin{lemma}[$AR(1)$ Bounds on Total Variation]
\label{lem:bnds}
Assuming that the output of the Markov chain follows the $AR(1)$ model above, we have the two bounds:
\[
\begin{split}
\sup_A|\bb P[Y_t\in A]-\bb P[Y_\infty\in A]|&\leq\\
 \textstyle\leq\sqrt{2-2\sqrt{\frac{2\sqrt{1-\rho^{2t}}}{2-\rho^{2t}}}\exp\left(-\frac{\left(c\frac{1-\rho^t}{1-\rho}+\rho^ty_0-\frac{c}{1-\rho}\right)^2}{4\sigma_\varepsilon^2\left(\frac{1-\rho^{2t}}{1-\rho^2}+\frac{1}{1-\rho^2}\right)}\right)},&\\
\sup_A|\bb P[Y_t\in A]-\bb P[Y_\infty\in A]|&\leq\\
 \textstyle\leq\frac{1}{2}\sqrt{\rho^{2t}\frac{(y_0-c/(1-\rho))^2}{\sigma_{\epsilon}^2/(1-\rho^2)}-\rho^{2t}-\ln(1-\rho^{2t})},&
\end{split}
\]
where the first one is derived using  the Hellinger distance and the second one is derived using  the Kullback-Leibler (KL) distance. 
\end{lemma}
In the following examples, we use estimates of $c,\rho$ for each $\beta_j,j=1,\ldots,p$ and plug them into the $AR(1)$-based upper bounds. 
We then use the largest of these $p$-estimates as 
a heuristic approximation of the true total variation distance of the Gibbs sampler. In a sense, this is equivalent to picking the autocorrelation plot that appears to decay at the slowest rate.  

\subsection{Diabetes dataset}
\label{sec:diabetes}
We present the result of our numerical study on the diabetes dataset of \cite{Efron2004}. The dataset consists of $10$ predictor variables (age, sex, BMI, etc.) and a response variable which is a medical measurement for the level of diabetes for $n=442$ patients. We model the variables using the Bayesian Lasso linear regression, and apply the regenerative Gibbs sampler to sample from the posterior distribution. 

\paragraph{Regenerative variance estimator \eqref{TAVC}.} The result is given in Table~\ref{fig:NumericResult}, in which we run the sampler to generate $5000$ samples and observed $3369$ regenerations. We also compare our regenerative estimator \eqref{TAVC} with the $AR(1)$-based estimates for the standard errors\cite{law1984confidence1,Gilks1998adaptive}.
\begin{table}[htb]
\begin{center}
\caption{A comparison of the standard error estimators for the diabetes dataset.}
\begin{scriptsize}
\begin{tabular}{c|c|c|c|c|c|}  
\cline{2 - 6}
& Mean & Standard error& $AR(1)$ st. err. & Relative error & $AR(1)$ rel. err.\\
\cline{2 - 6}
age &\num{-2.8513e+00} & \num{7.4526e-01} & \num{7.4691e-01} & \num{2.6145e-01} & \num{2.6203e-01} \\ 
\cline{2 - 6}
sex &\num{-2.1353e+02} & \num{9.0418e-01} & \num{9.0868e-01} & \num{4.2342e-03} & \num{4.2553e-03} \\ 
\cline{2 - 6}
bmi &\num{5.2416e+02} & \num{9.6080e-01} & \num{9.8536e-01} & \num{1.8331e-03} & \num{1.8799e-03} \\ 
\cline{2 - 6}
map &\num{3.0662e+02} & \num{9.3646e-01} & \num{9.8428e-01} & \num{3.0540e-03} & \num{3.2100e-03} \\ 
\cline{2 - 6}
tc  &\num{-1.9168e+02} & \num{2.9751e+00} & \num{3.4429e+00} & \num{1.5527e-02} & \num{1.7969e-02} \\ 
\cline{2 - 6}
ldl &\num{8.4696e+00} & \num{2.4224e+00} & \num{2.6600e+00} & \num{2.8794e-01} & \num{3.1618e-01} \\ 
\cline{2 - 6}
hdl &\num{-1.5009e+02} & \num{1.8383e+00} & \num{2.1300e+00} & \num{1.2246e-02} & \num{1.4190e-02} \\ 
\cline{2 - 6}
tch &\num{1.0027e+02} & \num{1.8644e+00} & \num{2.0679e+00} & \num{1.8594e-02} & \num{2.0624e-02} \\ 
\cline{2 - 6}
ltg &\num{5.2609e+02} & \num{1.5830e+00} & \num{1.7437e+00} & \num{3.0091e-03} & \num{3.3146e-03} \\ 
\cline{2 - 6}
glu &\num{6.4141e+01} & \num{8.8074e-01} & \num{9.2546e-01} & \num{1.3731e-02} & \num{1.4429e-02} \\ 
\cline{2 - 6}
\end{tabular}
\end{scriptsize}
\label{fig:NumericResult}
\end{center}
\end{table}


Both methods give estimates in the same ballpark.
It is also worthwhile noting that the $AR(1)$ approximation approach consistently gives larger standard error estimates than the regenerative approach.

\paragraph{Regenerative $\epsilon$-burn-in estimator \eqref{eps-burn-in}.} For the diabetes data set,  our estimate for $\eta$ is $1.0178$  ( $ 1.0178\pm 0.0008$ is a 95\% numerical confidence interval), therefore  an approximate $0.01$-burn-in period is $t_b\approx 101$.

For the $AR(1)$ approximation, substituting in the estimated parameters, we find that $t>5$ is sufficient for the $t$-th state to be within $0.01$ total variation distance to the stationary distribution. Thus, the $AR(1)$ approximation is very optimistic.
We note that we did not detect a practical difference between the two inequalities in Lemma~\ref{lem:bnds}. A comparison of  all bounds is given in Figure~\ref{fig:TVbounds}.
\begin{figure}[htb]
\begin{center}
\includegraphics[scale = 0.5]{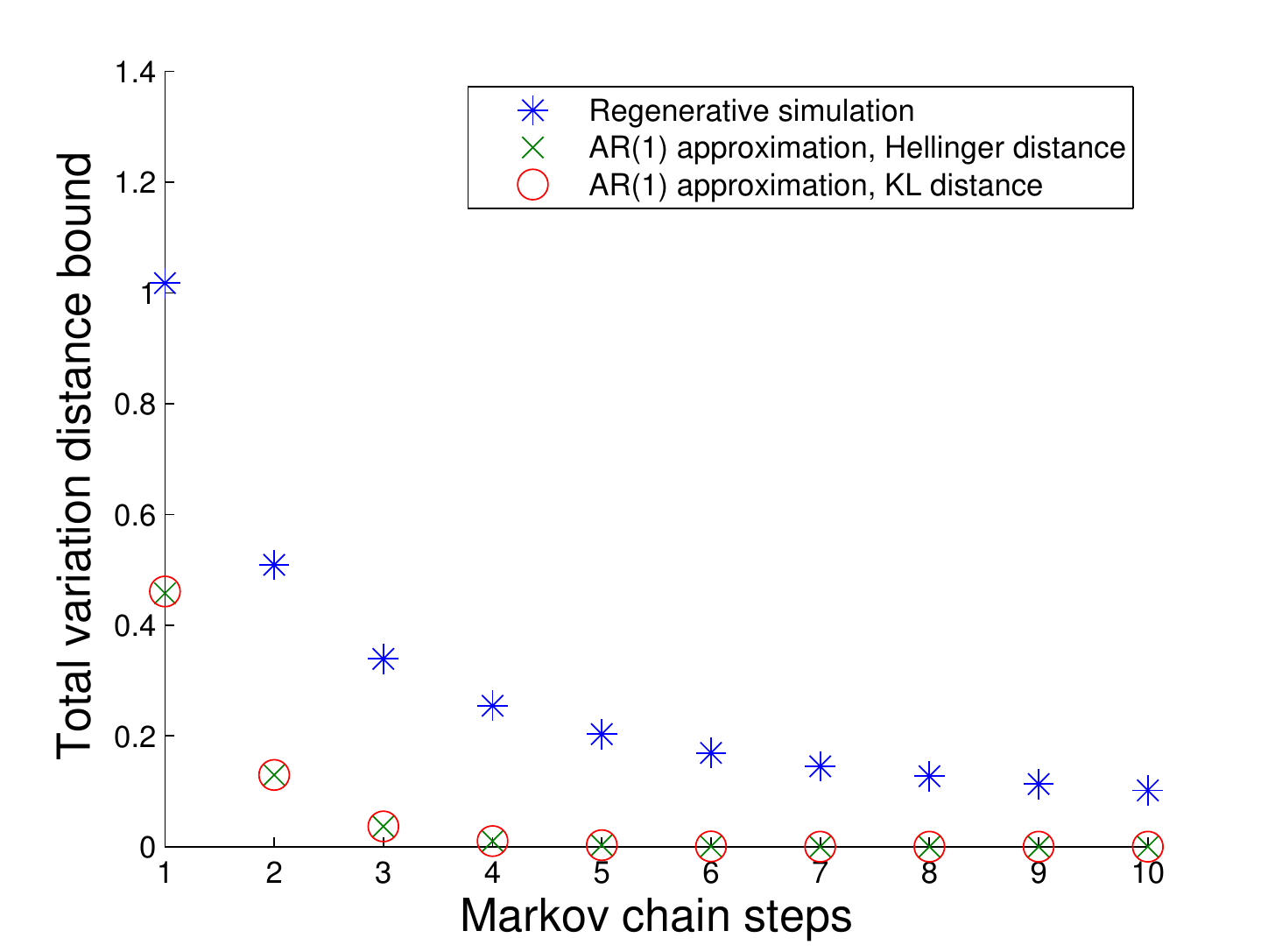}
\caption{A comparison of different approximations to the total variation distance bound.}
\label{fig:TVbounds}
\end{center}
\end{figure}

\subsection{Boston house price dataset}

The Boston house price dataset consists of $13$ predictor variables (crime rate per capita, proportion of residential land etc.) and a response variable which is the median value of owner-occupied homes for $n=506$ cases. We again model the variables using the Bayesian Lasso linear regression, and apply the regenerative Gibbs sampler to sample from the posterior distribution. 

\paragraph{Regenerative variance estimator \eqref{TAVC}.}
From Table~\ref{fig:NumericResult_Boston} we see that the regenerative variance estimator agrees with the $AR(1)$ approximation. It is worthwhile noting that the optimal average probability of regeneration is highly sensitive to the data. On the one hand, after optimizing with respect to $\alpha$, the diabetes dataset can achieve a probability of regeneration of more than $0.6$. On the other hand, the Boston house price dataset can barely achieve a probability of $0.04$. 
\begin{table}[htb]
\begin{scriptsize}
\begin{center}
\caption{A comparison of the standard error estimators for the Boston house price dataset}
\hspace{-1.5cm}\begin{tabular}{c|c|c|c|c|c|}  
\cline{2 - 6}
& Mean & Standard error& $AR(1)$ st. err. & Relative error & $AR(1)$ rel. err.\\
\cline{2 - 6}
crim &\num{-8.5303e-01} & \num{4.2883e-03} & \num{4.2029e-03} & \num{5.0241e-03} & \num{4.9241e-03} \\ 
\cline{2 - 6}
zn &\num{9.7795e-01} & \num{4.5787e-03} & \num{4.7987e-03} & \num{4.6810e-03} & \num{4.9059e-03} \\ 
\cline{2 - 6}
indus &\num{-1.1282e-03} & \num{5.7627e-03} & \num{5.4007e-03} & \num{4.1239e+00} & \num{3.8648e+00} \\ 
\cline{2 - 6}
chas &\num{6.8360e-01} & \num{3.3273e-03} & \num{3.1301e-03} & \num{4.8657e-03} & \num{4.5774e-03} \\ 
\cline{2 - 6}
nox &\num{-1.8889e+00} & \num{6.2934e-03} & \num{6.4078e-03} & \num{3.3304e-03} & \num{3.3909e-03} \\ 
\cline{2 - 6}
rm &\num{2.7099e+00} & \num{3.8381e-03} & \num{4.1590e-03} & \num{1.4167e-03} & \num{1.5351e-03} \\ 
\cline{2 - 6}
age &\num{-1.4436e-02} & \num{5.3901e-03} & \num{4.9713e-03} & \num{3.6201e-01} & \num{3.3388e-01} \\ 
\cline{2 - 6}
dis &\num{-2.9588e+00} & \num{5.4517e-03} & \num{5.9736e-03} & \num{1.8420e-03} & \num{2.0183e-03} \\ 
\cline{2 - 6}
rad &\num{2.2364e+00} & \num{8.4378e-03} & \num{9.1774e-03} & \num{3.7720e-03} & \num{4.1026e-03} \\ 
\cline{2 - 6}
tax &\num{-1.6823e+00} & \num{9.2717e-03} & \num{9.8996e-03} & \num{5.5113e-03} & \num{5.8845e-03} \\ 
\cline{2 - 6}
ptratio &\num{-2.0215e+00} & \num{3.9149e-03} & \num{4.0249e-03} & \num{1.9365e-03} & \num{1.9909e-03} \\ 
\cline{2 - 6}
b &\num{8.3368e-01} & \num{3.3051e-03} & \num{3.5869e-03} & \num{3.9689e-03} & \num{4.3073e-03} \\ 
\cline{2 - 6}
lstat &\num{-3.7205e+00} & \num{5.2659e-03} & \num{5.0817e-03} & \num{1.4153e-03} & \num{1.3658e-03} \\ 
\cline{2 - 6}
\end{tabular}
\label{fig:NumericResult_Boston}
\end{center}
\end{scriptsize}
\end{table}

\begin{figure}[htb]
\begin{center}
\label{fig:GridSearch_Boston}
\includegraphics[scale = 0.54]{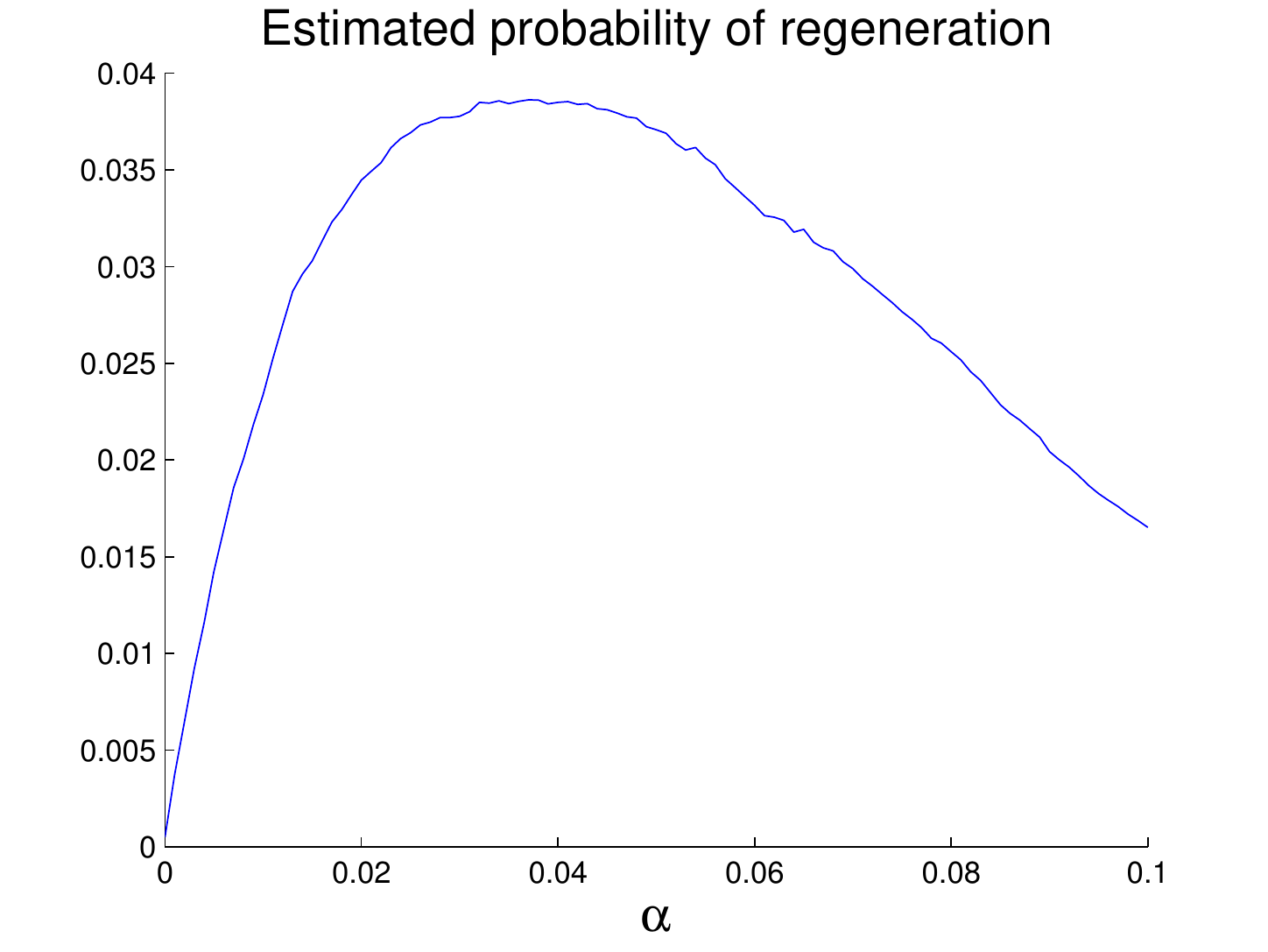}
\caption{Grid search to maximize the number of regenerations. Here the empirical Bayes estimator for the pair $(\lambda,\sigma)$ is $(0.613,4.68)$. We can see a clearly pronounced maximum at $\alpha\approx 0.04$.}
\end{center}
\end{figure}
Further experiments suggest that the dataset affects the probability of regeneration mainly through the estimated value for the Lasso parameter $\lambda$. In other words,  the number  of detected
regenerative events depends on
 the value of the Lasso parameter $\lambda$.  Although,  
fewer observed regenerations do not necessarily signify a high mean squared error, we need to observe at least two regenerations to be able to compute a valid estimate of the asymptotic variance.
Thus, a limitation of our regenerative sampling is that when
 $\lambda$ is very large, one may need to run the Markov chain for  many steps.

\paragraph{Regenerative $\epsilon$-burn-in estimator \eqref{eps-burn-in}.} 
 
For the housing data set,  our estimate for $\eta$ is $51$  ( $ 51\pm 6.1$ for a 95\% numerical confidence interval). Therefore,  an approximate $0.01$-burn-in period is $t_b\approx 5100$.

In contrast, the $AR(1)$ approximation is remarkably optimistic, as seen from this table of estimated  bounds:

\begin{center}
{\scriptsize
\begin{tabular}{c|c|c|c|c|c}
 distance bnd. & $9.5 \times 10^{-3}$ &  $7.8\times 10^{-4}$ &  
$1.1\times 10^{-4}$  & $1.7\times 10^{-5}$                  \\
\hline
step $t$ & 1 & 2 & 3& 4 
 \end{tabular}
}
\end{center}
We  can see   that $t\geq 1$ is sufficient for the $t$-th state to be within $0.01$ total variation distance to the stationary distribution. Thus, in this example there is a  significant disagreement between the regenerative  and the heuristic 
convergence assessment. We believe that while the regenerative estimate is too conservative, the heuristic one is too optimistic, and that the true $\epsilon$-burn-in is somewhere in-between. 

\section{Concluding Remarks}
In this paper, we  identify the regenerative times in the output of the popular Park \& Casella Gibbs sampler, which (approximately) simulates from the posterior of the Bayesian   Lasso. The resulting regenerative simulation algorithm allows practitioners to answer the two key questions that need to be answered for any  convergence assessment \cite{jones2001honest} of a Markov chain:
\begin{enumerate}[label=(\roman*)]
	\item What is the statistical error of any empirical average that aims to estimate a stationary quantity of interest? The answer is provided by the consistent mean squared error estimator \eqref{TAVC}.
	\item How long does it take for the Markov chain to get sufficiently close to the limiting distribution? One good answer is the consistent estimator \eqref{eps-burn-in} of 
	(an upper bound on)  the	$\epsilon$-burn-in of the Markov chain.
\end{enumerate}
Whenever applicable, one should use the regenerative estimators to tackle issue (i) and (ii), because  the popular alternatives, such as batch-means estimators and diagnostic plots, do not have the same sound theoretical foundation for their use.

While the  estimator \eqref{eps-burn-in} of the mean square error is not  novel   \cite{mykland1995regeneration,jones2001honest},   the regenerative estimator \eqref{eps-burn-in} of the burn-in period  appears to have been overlooked as a more rigorous approach to estimating the burn-in. This article fills this gap.

%

\appendix
\section{Proof of theorem~\ref{th:converge}}

We use the notation from Section~\ref{sec:reg process}. Let $\bb Q_t[A]$ be the distribution of a state $\v X$ 
picked  at random from the Markov chain states: $\v X_0,\ldots,\v X_t$. In other words,
\[
\bb Q_t[A]:=
\frac{1}{t+1}\sum_{k=0}^{t} \bb E[\kappa_k(A\gvn \v X_0)],
\]
where $\bb E[\kappa_k(A\gvn \v X_0)]=\bb P[\v X_k\in A]$ for all $k\geq 0$.
By assumption, the Markov chain is geometrically ergodic, that is, the distribution of the length $M$ of a  regenerative cycle of the chain is light-tailed. In other words,
$\bb E\exp(\epsilon_1 M)<\infty$ for some $\epsilon_1>0$. The process $\v X_0,\v X_1,\ldots$ is also a zero-delayed regenerative process, because by assumption the initial $\v X_0$ commences a new cycle. Therefore, the conditions of   Lemma~\ref{lem:uniform bias} (see below) are satisfied and we have:
\begin{equation}
\label{eq1}
\left\|\bb Q_t-\pi\right\|_\mathrm{TV}\leq \frac{\bb E M^2-\bb EM}{2(t+1)\bb EM}+\c O(\exp(-\varepsilon t))
\end{equation}
for some $\epsilon\in(0,\epsilon_1]$. In addition, Lemma~\ref{lem:time-avg bound} below states  that the distribution of the final state of the Markov chain, $\v X_t$, is closer to $\pi$ than a state picked at random from the history of the chain up until time $t$:  $\v X_0,\ldots,\v X_t$. In other words,
\begin{equation}
\label{eq2}
\|\bb E[\kappa_t(\cdot\gvn \v X_0)]-\pi\|_\mathrm{TV}\leq \left\|\bb Q_t-\pi\right\|_\mathrm{TV}
\end{equation}
The result of the theorem then follows by combining \eqref{eq1}+\eqref{eq2}. 
\begin{lemma}[Uniform bias estimate]
\label{lem:uniform bias}
Suppose $\v X_0,\v X_1,\ldots$ is a zero-delayed discrete regenerative process with 
regeneration times $0=T_0<T_1<T_2<\cdots$, where $T_n=M_1+\cdots+M_n$, and stationary distribution $\bb Q$. 
Let $\bb E \exp(\varepsilon_1 M)<\infty$ for some $\varepsilon_1> 0$ and let $\bb Q_t$
be the distribution of a state drawn at random from the whole history of the chain up until time $t$, that is, drawn at random from $\v X_0,\ldots,\v X_t$. Then, we have for some $\varepsilon\in(0,\varepsilon_1]$
\[
\sup_A|\bb Q_t[A]-\bb Q[A]|\leq \frac{\bb E M^2-\bb EM}{2(t+1)\bb EM}+\c O(\exp(-\varepsilon t))
\]
\end{lemma}
\begin{proof}
The proof follows closely the ideas in \cite{glynn1994some}. Using the notation from Section~\ref{sec:reg process}, let 
$u(k)=\sum_{j=0}^k\bb P(T_j=k)=\bb P(\exists j:T_j=k)
$
 denote the renewal measure, and define the convolution operator 
$(a* b)(t)=\sum_{k=0}^t a(t-k)b(k)$ between two functions $a$ and $b$.
Further, define 
\[
\begin{split}
e_A(t)&:=(t+1)
(\bb Q[A]-\bb Q_t[A])=\bb E\sum_{k=0}^{t} Z_k(A),
\end{split}
\]
where $Z_k(A)=\bb Q[A]-\bb I\{\v X_k\in A\} $.
Wald's identity implies that
\[
\bb E \sum_{k=0}^{M_1-1}Z_k(A)=\bb Q[A]\bb EM_1-\bb E H_1(A)=0.
\]
Thus, we can then verify that $e_A$ satisfies the renewal equation
\[
e_A(t)=(v_A*u)(t),
\]
 where
\[
\begin{split}
v_A(t)&:=\textstyle\bb E\left[\sum_{k=0}^{M_1-1} Z_k(A)-\sum_{k=0}^{t} Z_k(A) ;M_1>t\right]\\
&=\textstyle\bb E\left[\sum_{k=t+1}^{M_1-1} Z_k(A) ;M_1\geq t+2\right]
\end{split}
\] 
with
\[
|v_A(t)|\leq \bb E\left[M_1-t-1 ;M_1\geq t+2\right]
\]
Since $\bb E\exp(\varepsilon_1 M)<\infty$, then there exists some
$\varepsilon_2\in(0,\varepsilon_1]$ such that 
$\bb EM\exp(\varepsilon_2 M)\leq \infty$, and
therefore
\[
\begin{split}
|v_A(t)|\leq \bb E\left[M;M\geq t\right]&\leq \frac{\bb E[M\exp(\varepsilon_2 M);M\geq t]}{\exp(\varepsilon_2 t)}\\
&\leq \frac{\bb E[M\exp(\varepsilon_2 M)]}{\exp(\varepsilon_2 t)}=\c O(\exp(-\varepsilon_2 t))
\end{split}
\]
An application of \cite[Theorem 2.10 on Page 196]{asmussen2008applied} yields  for some $\varepsilon\in(0,\varepsilon_2]$:
\[
e_A(t)= \frac{\sum_{k\geq 0} v_A(k)}{\bb EM}+\c O(\exp(-\varepsilon t)),\quad t\uparrow\infty
\]
uniformly in $A$. In other words,
\[
\sup_A|e_A(t)|\leq \frac{\sum_{k\geq 0} \sup_A|v_A(k)|}{\bb EM}+\c O(\exp(-\varepsilon t)),\quad t\uparrow\infty
\]
Simplifying the upper bound $\sum_{k\geq 0} \sup_A|v_A(k)|\leq \sum_{k\geq 0}\bb E[(M-k-1)^+]=\frac{\bb EM^2-\bb EM}{2}$ yields the desired result.

\end{proof}

\begin{lemma}[Time-average bound on total variation distance]
\label{lem:time-avg bound}
Let $\v X_1,\v X_2,\ldots,\v X_t$ be a Markov chain with a $t$-step transition kernel $\kappa_t(\cdot\,|\,  \v X_0)$, $\v X_0\sim \pi_0$ and a stationary/limiting distribution $\pi$. Then, the distribution of a 
random variable drawn uniformly from 
$\v X_1,\ldots,\v X_t$ is further away from $\pi$
than the distribution of the last state $\v X_t$. In other words, we have  that 
\[
\sup_{A}|\bb E[\kappa_t(A\gvn \v  X_0)]-\pi(A)|\leq 
\sup_{A}\Big|\frac{1}{t}\sum^{t}_{j=1}(\bb E[\kappa_j(A\gvn \v  X_0)]-\pi(A))\Big|
\]
\end{lemma}
\begin{proof}
For simplicity, assume that there exist  densities
 $\bb E[\kappa_t( \v x\gvn  \v X_0)]$ and $\pi( \v x)$, corresponding to  $\bb E[\kappa_t(A\gvn  \v X_0)]=\bb P[\v X_t\in A]$ and $\pi(A)$ (which is the case with the Gibbs sampler of Park \& Casella anyway). We then use the following four facts.

First, for any nonnegative functions $g_1$ and $g_2$, we have
$
\min\{g_1(\v x)+g_2( \v x),1\}\leq \min\{g_1(\v x),1\}+\min\{g_2( \v x),1\}
$.
More generally, for any nonnegative functions $g_1,g_2,\ldots$
\begin{equation}
\label{ineq min}
\min\Big\{\sum_jg_j(\v x),1\Big\}\leq \sum_j\min\{g_j( \v x),1\}
\end{equation}
Second,  Sheffe's lemma states that for any densities $p$ and $q$:
\begin{equation}
\label{Sheffe2}
\begin{split}
\sup_A\left|\int_A p(\v x)\di \v x-\int_A q(\v x)\di \v x\right|&=\frac{1}{2}\int |p(\v x)-q(\v x)|\di \v x\\
&=1-\int \min\{p(\v x),q( \v x)\}\di \v x
\end{split}
\end{equation}
Third, for any $s'\leq s$, we have
\begin{equation}
\label{monotonicity}
\|\bb E[\kappa_{s'}(\cdot\gvn  \v X_0)]-\pi\|_\mathrm{TV}\geq
\|\bb E[\kappa_{s}(\cdot\gvn \v X_0)]-\pi\|_\mathrm{TV},
\end{equation}
which is nothing more than a statement of the obvious fact that the more we run the Markov chain, the closer we get to its stationary distribution. 
Fourth, two random variables $\v X_1$ and $\v X_2$ on the same probability space and with marginal distributions $\pi_1(A)=\bb P[\v X_1\in A]$ and $\pi_2(A)=\bb P[\v X_2\in A]$
are  maximally coupled \cite{thorisson1986maximal} when their joint distribution is such that
\[
\|\pi_1-\pi_2\|_\mathrm{TV}=\bb P[\v X_1\not =\v X_2]
\]
We now apply these four results as follows.
Let $\v X$ and $\v Y$ be maximally coupled with marginal densities
$\frac{1}{t}\sum^{t}_{j=1}\bb E[\kappa_j(\v x\gvn  \v X_0)]$ and $\pi(\v y)$, respectively. Then, $\bb P[\v X\not = \v Y]=\sup_{A}\left|\frac{1}{t}\sum^{t}_{j=1}(\bb E[\kappa_j(A\gvn \v X_0)]-\pi(A))\right|$ and
\[
\begin{split}
\bb P[\v X=\v Y]&\stackrel{\eqref{Sheffe2}}{=}\int \pi(\v x) \min\Big\{\frac{1}{t}\sum^{t}_{j=1}\frac{\bb E[\kappa_j(\v x\gvn \v  X_0)]}{\pi(\v x)},1\Big\} \di \v x\\
&\stackrel{\eqref{ineq min}}{\leq} \frac{1}{t}\sum^{t}_{j=1}\int\min\left\{\bb E[\kappa_j(\v x\gvn \v X_0)],\pi(\v x)\right\} \di\v x\\
&\stackrel{\eqref{Sheffe2}}{=} \frac{1}{t}\sum^{t}_{j=1}
\left(1-\sup_A|\bb E[\kappa_j(A\gvn  \v X_0)]-\pi(A)|\right)\\
&= 1-\frac{1}{t}\sum^{t}_{j=1}\|\bb E[\kappa_j(\cdot\gvn  \v X_0)]-\pi\|_\mathrm{TV}.
\end{split}
\]
Then, using $\bb P[\v X\not= \v Y]=1-\bb P[\v X=\v Y]$, the last inequality implies that
\[
\bb P[\v X\not= \v Y]\geq \frac{1}{t}\sum^{t}_{j=1}\|\bb E[\kappa_j(\cdot\gvn  \v X_0)]-\pi\|_\mathrm{TV}\stackrel{\eqref{monotonicity}}{\geq} \|\bb E[\kappa_t(\cdot\gvn  \v X_0)]-\pi\|_\mathrm{TV},
\]
whence the desired result follows. 
\end{proof}

\section{Background: Gibbs sampler for the Bayesian Lasso}
\label{app:BLasso}

The first key insight in \cite{Park2008bayesian} is that a $\mathsf{Laplace}(0,1/\lambda)$ density is in fact a Gaussian-scale mixture \cite{Andrews1974}. In particular, for each $\beta_j$, $j\in\{1,\ldots,p\}$, we have the identity,
\[
\frac{\lambda}{2}\exp\left(-\lambda|\beta_j|\right)=\int_0^\infty\frac{1}{\sqrt{2\pi s_j}}\exp\left(-\frac{\beta_j^2}{2s_j}\right)\times\frac{\lambda^2}{2}\exp\left(-\frac{\lambda^2}{2}s_j\right)\;\di  s_j.
\]
It follows form the change of variable $\tau_j=1/s_j$,
\begin{equation}
\frac{\lambda}{2}\exp\left(-\lambda|\beta_j|\right)=\int_0^\infty
\frac{\lambda^2}{2\sqrt{2\pi}}\exp\left(-\frac{\beta_j^2\tau_j}{2}\right)\exp\left(-\frac{\lambda^2}{2\tau_j}\right)\tau_j^{-3/2}\;\di \tau_j.
\label{eq:identity}
\end{equation}
Hence if one considers sampling the pair $(\v\beta,\v\tau)\in\bb R^p\times\bb R_+^p$ from the joint density $\pi(\v\beta,\v\tau|\lambda,\sigma^2)=$
\begin{equation}
=\frac{\exp\left(-\frac{1}{2\sigma^2}\|\v y-\m X\v\beta\|_2^2\right)\prod_{j=1}^p\frac{\lambda^2}{2\sqrt{2\pi}}\exp\left(-\frac{\beta_j^2\tau_j}{2}\right)\exp\left(-\frac{\lambda^2}{2\tau_j}\right)\tau_j^{-3/2}}{\ell(\lambda,\sigma^2)},
\label{eq:target}
\end{equation}
the marginal samples $\v\beta$, from samples of the pair $(\v\beta,\v\tau)$, have the same distribution as \eqref{eq:posterior}. This is because \eqref{eq:identity} implies $\int_{\bb R_+^p}\pi(\v\beta,\v\tau|\lambda,\sigma^2)\;\di \v\tau=\pi(\v\beta|\lambda,\sigma^2)$.

The form of \eqref{eq:target} suggests a natural (block) Gibbs sampler that cycles between the full conditional distributions $\pi(\v\beta|\v\tau,\lambda,\sigma^2)$ and $\pi(\v\tau|\v\beta,\lambda,\sigma^2)$. The second key insight in \cite{Park2008bayesian} is that $\pi(\v\tau|\v\beta,\lambda,\sigma^2)$ takes the product form
\[
\pi(\v\tau|\v\beta,\lambda,\sigma^2)\varpropto\prod_{j=1}^p\exp\left(-\frac{\beta_j^2\tau_j}{2}\right)\exp\left(-\frac{\lambda^2}{2\tau_j}\right)\tau_j^{-3/2}.
\]
This means each $\tau_j$ are conditionally independent. Moreover, the conditional distribution of $\tau_j$ is $\mathsf{Wald}(\lambda',\mu_j')$ where $\lambda'=\lambda^2$ and $\mu_j'=\lambda/|\beta_j|$ (see, for example, \cite{chhikara1988}). Finally, it is not hard to show that
\[
\pi(\v\beta|\v\tau,\lambda,\sigma^2)\varpropto\exp\left(-\frac{1}{2\sigma^2}\left(\v\beta-\m A\m X^\top \v y \right)^\top\m A^{-1}\left(\v\beta-\m A\m X^\top \v y\right)^\top\right),
\]
where $\m A^{-1}=\m X^\top\m X+\mbox{diag}(\v\tau)$ is a symmetric invertible matrix. This means, $\v\beta$ conditional on $\v\tau$, is a $p$-dimensional Gaussian random variable with the mean vector $\m A\m X^\top \v y$ and the covariance matrix $\sigma^2\m A$.

At this stage one may wonder how  we determine the pair $(\lambda,\sigma^2)$. In fact, one may choose to adapt a fully Bayesian approach and assign the pair some prior distributions, see \cite{Leng2014}. However, in this paper we take the empirical Bayes approach and use the estimator $(\hat\lambda,\hat\sigma^2)=\mbox{argmax }\ell(\lambda,\sigma^2)$. This is because the parameters $(\lambda,\sigma^2)$ are rarely of interest, that is, they are nuisance parameters. 
In this paper, we use the approximate EM algorithm of \cite{Casella2001} to solve the program $(\hat\lambda,\hat\sigma^2)=\mbox{argmax }\ell(\lambda,\sigma^2)$.

\section{Proof of Lemma~\ref{lem:minorization}}
Our strategy follows from the approach described in \cite{mykland1995regeneration} and used in \cite{roy2007convergence}. Denote $\mathscr{X}=\bb R^p\times\bb R_+^p$ and fix $(\tilde{\v\beta},\tilde{\v\tau})\in\mathscr{X}$ and $\c D\subseteq\mathscr {X}$. Observe that $ \kappa(\v\beta_{k+1},\v\tau_{k+1}|\v\beta_k,\v\tau_k)=$
\[
 \begin{split}
&=\pi(\v\tau_{k+1}|\v\beta_k)\pi(\v\beta_{k+1}|\v\tau_{k+1})\\
 &=\frac{\pi(\v\tau_{k+1}|\v\beta_k)}{\pi(\v\tau_{k+1}|\tilde{\v\beta})}\pi(\v\tau_{k+1}|\tilde{\v\beta})
 \pi(\v\beta_{k+1}|\v\tau_{k+1})\\
 &\geq\varepsilon\inf_{(\v\beta_{k+1},\v\tau_{k+1})\in \c D}\left\{\frac{\pi(\v\tau_{k+1}|\v\beta_k)}
 {\pi(\v\tau_{k+1}|\tilde{\v\beta})}\right\}\varepsilon^{-1}\pi(\v\tau_{k+1}|\tilde{\v\beta})
 \pi(\v\beta_{k+1}|\v\tau_{k+1})\bb I\left\{(\v\beta_{k+1},\v\tau_{k+1})\in D\right\}\\
 &=\varepsilon\inf_{(\v\beta_{k+1},\v\tau_{k+1})\in \c D}\left\{\frac{\pi(\v\tau_{k+1}|\v\beta_k)}
 {\pi(\v\tau_{k+1}|\tilde{\v\beta})}\right\}\varepsilon^{-1}\kappa(\v\beta_{k+1},\v\tau_{k+1}|\tilde{\v\beta},\tilde{\v\tau})
 \bb I\left\{(\v\beta_{k+1},\v\tau_{k+1})\in \c D\right\}
 \end{split}
\]
where
\[
 \varepsilon=\int_{\c D}\kappa (\v\beta_{k+1},\v\tau_{k+1}|\tilde{\v\beta},\tilde{\v\tau}) \;\di (\v\beta_{k+1},\v\tau_{k+1}).
\]
In particular, let us take $\c D=\bb R^p\times\prod_{j=1}^p[c_j,d_j]:=\bb R^p\times [\v c,\v d]$ and $(\tilde{\v\beta},\tilde{\v\tau})=(\hat{\v\beta},\v1)$ where 
$\hat{\v\beta}$ is the solution to the frequentist Lasso penalized regression model. It follows that we can
take
\[
 \begin{split}
 \varepsilon^{-1}s(\v\beta_{k},\v\tau_k)&=\inf_{(\v\beta_{k+1},\v\tau_{k+1})\in\c D}\left\{\frac{\pi(\v\tau_{k+1}|\v\beta_k)}
 {\pi(\v\tau_{k+1}|\hat{\v\beta})}\right\}\\
 &=\inf_{\v\tau_{k+1}\in[\v c,\v d]}\exp\left(-\frac{1}{2}\sum_{j=1}^{p}\tau_{k+1,j}(\beta_{k,j}^2-\hat\beta^2_{j})\right)\\
 \varepsilon^{-1}s(\v\beta_{k},\v\tau_k)&=\textstyle\exp\left(-\frac{\sum_{j=1}^{p}d_j
 (\beta_{k,j}^2-\hat\beta^2_{j})\bb I\left\{\beta_{k,j}^2-\hat\beta^2_{j}\geq 0\right\}}{2}-\frac{\sum_{j=1}^{p}c_j(\beta_{k,j}^2-\hat\beta^2_{j})\bb I\left\{\beta_{k,j}^2-\hat\beta^2_{j}< 0\right\}}{2}\right)\\
&=\textstyle \exp\left(-\frac{\v d^\top\v(\v\beta_k^2-\hat{\v\beta}^2)_+}{2}-\frac{\v c^\top\v(\v\beta_k^2-\hat{\v\beta}^2)_-}{2}\right)
 \end{split}
\]
Here $j$ is the index for entries within the vectors and $k$ is the index for the steps
in the Markov chain. The above calculation recalls the fact that the normalizing constant for the density function of a
$\mathsf{Wald}(\lambda',\mu')$ random variable is $(\lambda/2\pi)^{1/2}$. 

Denote
\[
 \nu(\v\beta_{k+1},\v\tau_{k+1})=\varepsilon^{-1}\kappa(\v\beta_{k+1},\v\tau_{k+1}|\tilde{\v\beta}
 ,\tilde{\v\tau}) 
 \bb I\left\{(\v\beta_{k+1},\v\tau_{k+1})\in\c D\right\}.
\]
Therefore, by construction we have
\[
\kappa((\v\beta_{k+1},\v \tau_{k+1})|(\v \beta_k,\v\tau_k))\geq\nu(\v\beta_{k+1},\v\tau_{k+1})s(\v\beta_k,\v\tau_k),\qquad\forall (\v\beta_k,\v\tau_k)
\]
as required. 
For the probability of regeneration, we then obtain:
	\[
	 \begin{split}
	  \psi_k&=\frac{s(\v\beta_k,\v\tau_k)\nu(\v\beta_{k+1},\v\tau_{k+1})}{\kappa(\v\beta_{k+1},\v\tau_{k+1}|\v\beta_{k},\v\tau_{k})}\\
		&=\frac{\varepsilon^{-1}s(\v\beta_k,\v\tau_k)\pi(\v\tau_{k+1}|\tilde{\v\beta})
 \pi(\v\beta_{k+1}|\v\tau_{k+1})\bb I\left\{(\v\beta_{k+1},\v\tau_{k+1})\in D\right\}}{\pi(\v\tau_{k+1}|\v\beta_t)\pi(\v\beta_{k+1}|\v\tau_{k+1})}\\
		&=\frac{\varepsilon^{-1}s(\v\beta_k,\v\tau_k)\pi(\v\tau_{k+1}|\tilde{\v\beta})\bb I\left\{\v\tau_{k+1}\in[\v c,\v d]\right\}}{\pi(\v\tau_{k+1}|\v\beta_k)}\\
		&=\varepsilon^{-1}s(\v\beta_k,\v\tau_k)\exp\left(\frac{1}{2}\sum_{j=1}^{p}\tau_{k+1}\left(\beta_{k,j}^2-\hat\beta_{j}^2\right)\right)\times\bb I\left\{\v\tau_{k+1}\in[\v c,\v d]\right\}\\
		&=\exp\Bigg(-\frac{1}{2}\sum_{j=1}^{p}(u_j-\tau_{k+1,j})
 \left(\beta_{k,j}^2-\hat\beta^2_{j}\right)\bb I\left\{\beta_{k,j}^2-\hat\beta^2_{j}\geq 0\right\}\\
 &\quad-\frac{1}{2}\sum_{j=1}^{p}(l_j-\tau_{k+1,j})\left(\beta_{k,j}^2-\hat\beta_{j}^2\right)\bb I\left\{\beta_{k,j}^2-\hat\beta^2_{j}< 0\right\}\Bigg)\times \bb I\left\{\v\tau_{k+1}\in[\v c,\v d]\right\}
		\end{split}
	\]
	Therefore, 
	\begin{equation*}
	\psi_k=\exp\left(-\frac{1}{2}(\v u-\v\tau_{k+1})^\top(\v\beta_k^2-\hat{\v\beta}^2)_+-\frac{1}{2}(\v l-\v\tau_{k+1})^\top(\v\beta_k^2-\hat{\v\beta}^2)_-\right)\times \bb I\left\{\v\tau_{k+1}\in[\v c,\v d]\right\}.
	\end{equation*}

\section{Proof of Lemma~\ref{lem:bnds}}

The first  bound is derived
from the facts: i) 
$\frac{1}{2}\|\phi_1-\phi_2\|_1\leq  \|\sqrt{\phi_1}-\sqrt{\phi_2}\|_2$, where $\phi_1$ and $\phi_2$ are any probability densities; and  ii)
for two Gaussian densities $\phi_1, \phi_2$ with means $(\mu_1,\mu_2)$
and variances $(\sigma_1^2,\sigma_2^2)$ the $L_2$ norm $\|\sqrt{\phi_1}-\sqrt{\phi_2}\|_2^2$ is explicitly given by $2-2\sqrt{\frac{2\sigma_1\sigma_2}{\sigma^2_1+\sigma_2^2}}\exp\left(-\frac{(\mu_1-\mu_2)^2}{4(\sigma_1^2+\sigma_2^2)}\right)$. Then, we obtain
\[\textstyle
\sup_A|\bb P[Y_t\in A]-\bb P[Y_\infty\in A]|\leq \sqrt{2-2\sqrt{\frac{2\sigma_t\sigma}{\sigma^2_t+\sigma^2}}\exp\left(-\frac{(\mu_t-\mu)^2}{4(\sigma_t^2+\sigma^2)}\right)},
\]
where $Y_t$ is a Gaussian with 
mean $c\frac{1-\rho^t}{1-\rho}+\rho^t y_0$ and variance $\sigma^2_\varepsilon\frac{1-\rho^{2t}}{1-\rho^2}$. 
The second, slightly looser bound, is obtained via Pinsker inequality.

\bibliographystyle{plain}

\end{document}